\documentclass[journal]{IEEEtran}
\usepackage{subfig}
\usepackage{graphicx}
\usepackage{amsthm,amsfonts,amssymb,amsmath}
\usepackage{algorithm}
\usepackage{algpseudocode}
\usepackage{color}
\usepackage{bm,bbm}

\newtheorem{thm}{Theorem}
\newtheorem{rmk}{Remark}
\newtheorem{lemma}{Lemma}
\newtheorem{defi}{Definition}

\newtheorem{cor}{Corollary}

\usepackage{color}
\definecolor{mygrey}{gray}{0.50}

\newcommand{\II}{\mathbf{I}}

\newcommand{\eff}{\text{eff}}

\newcommand{\xx}{\mathbf{x}}
\newcommand{\yy}{\mathbf{y}}

\newcommand{\uu}{\mathbf{u}}
\newcommand{\ww}{\mathbf{w}}
\newcommand{\WW}{\mathbf{W}}
\newcommand{\HH}{\mathbf{H}}
\newcommand{\EE}{\mathbf{E}}
\newcommand{\UU}{\mathbf{U}}
\newcommand{\BB}{\mathbf{B}}
\newcommand{\YY}{\mathbf{Y}}

\newcommand{\XX}{\mathbf{X}}
\newcommand{\FF}{\mathbf{F}}
\newcommand{\RR}{\mathbf{R}}
\newcommand{\SSigma}{\mathbf{\Sigma}}
\newcommand{\nozero}{\backslash\left\{\mathbf{0}\right\}}

\begin{document}

\title{Semantically Secure Lattice Codes \\for Compound MIMO Channels}

\author{Antonio Campello, Cong Ling and Jean-Claude Belfiore\thanks{

This work was presented in part at the International Zurich Seminar on Communications (IZS) 2018 and in part at the International Symposium on Turbo Codes and Iterative Information Processing (ISTC) 2016.

A. Campello is with the Department of Electrical and Electronic Engineering, Imperial College London, London SW7 2AZ, U.K. (e-mail: a.campello@imperial.ac.uk).

C.  Ling  is  with  the  Department  of  Electrical  and  Electronic  Engineering, Imperial College London, London SW7 2AZ, U.K. (e-mail: cling@ieee.org).

J.-C.   Belfiore    is    with the  Mathematical and Algorithmic Sciences Lab, France Research Center,
Huawei Technologies   (e-mail: belfiore@telecom-paristech.fr).

}}

\maketitle
\begin{abstract}
We consider compound multi-input multi-output (MIMO) wiretap channels where minimal channel state information at the transmitter (CSIT) is assumed. Code construction is given for the special case of isotropic mutual information, which serves as a conservative strategy for general cases. Using the flatness factor for MIMO channels, we propose lattice codes universally achieving the secrecy capacity of compound MIMO wiretap channels up to a constant gap (measured in nats) that is equal to the number of transmit antennas. The proposed approach improves upon existing works on secrecy coding for MIMO wiretap channels from an error probability perspective, and establishes information theoretic security (in fact semantic security). We also give an algebraic construction to reduce the code design complexity, as well as the decoding complexity of the legitimate receiver. Thanks to the algebraic structures of number fields and division algebras, our code construction for compound MIMO wiretap channels can be reduced to that for Gaussian wiretap channels, up to some additional gap to secrecy capacity.
\end{abstract}

\section{Introduction}

Due to the open nature of the wireless medium, wireless communications are inherently vulnerable to eavesdropping attacks.
Information theoretic security offers additional protection for wireless data, since it only relies on the physical properties of wireless channels, thus representing a competitive/complementary approach to security compared to traditional cryptography.

The fundamental wiretap channel model was first introduced by Wyner \cite{Wyner75}. In this seminal paper, Wyner defined the secrecy capacity and presented the idea of coset coding to encode both data and random bits to mitigate eavesdropping. In recent years, the quest for the secrecy capacity of many classes of channels has been one of the central topics in wireless communications \cite{BBRM08,Bloch_Barros_2011,Liang_Poor_Shamai_2009,polarsecrecy,BargWiretap,GeneralWiretap,Tyagi15}.

In the information theory community, a commonly used secrecy notion is \emph{strong secrecy}: the mutual information $\mathbb{I}(M;Z^T)$ between the confidential message $M$ and the channel output $Z^T$ should vanish when the code length $T \to \infty$. This common assumption of uniformly distributed messages was relaxed in \cite{Bellare2012}, which considered the concept of \emph{semantic security}: for \textit{any} message distribution, the advantage obtained by an eavesdropper from its received signal vanishes for large block lengths. This notion is motivated by the fact that the plaintext can be fixed and arbitrary.

For the Gaussian wiretap channel, \cite{OSB} introduced the \emph{secrecy gain} of lattice codes while \cite{LLBS_12} proposed semantically secure lattice codes based on the lattice Gaussian distribution. To obtain semantic security, the \emph{flatness factor} of a lattice was introduced in \cite{LLBS_12} as a fundamental criterion which implies that conditional outputs are indistinguishable for different input messages. Using a random coding argument, it was shown that there exist families of lattice codes which are \emph{good for secrecy}, meaning that their flatness factor vanishes. Such families achieve semantic security for rates up to $1/2$ nat from the secrecy capacity.

Compared to the Gaussian wiretap channel, the cases of fading and multi-input multi-output (MIMO) wiretap channels are more technically challenging. The fundamental limits of fading wireless channels with secrecy constraints have been investigated in \cite{BR06,LPS07,BBRM08}, where the achievable rates and the secrecy outage probability were given. The secrecy capacity of the MIMO wiretap channel was derived in \cite{OH11,Khisti10,LiuShamai09,Loyka16Wiretap}, assuming full channel state information at the transmitter (CSIT). A code design in this setting was given in \cite{KhinaKK15} by reducing to scalar Gaussian codes. Although CSIT is sometimes available for the legitimate channel, it is hardly possible that it would be available for the eavesdropping channel. An achievability result was given in \cite{He14} for varying MIMO wiretap channels with no CSI about the wiretapper, under the condition that the wiretapper has less antennas than the legitimate receiver. Schaefer and Loyka \cite{SL15} studied the secrecy capacity of the \textit{compound} MIMO wiretap channel, where a transmitter has no knowledge of the realization of the eavesdropping channel, except that it remains fixed during the transmission block and belongs to a given set (the \textit{compound set}). The compound model represents a well-accepted reasonable approach to information theoretic security, which assumes minimal CSIT of the eavesdropping channel \cite{Liang09,Bjelakovic2013,Khisti11}. It can also model a multicast channel with several eavesdroppers, where the transmitter sends information to all legitimate receivers while keeping it secret from all eavesdroppers \cite{Liang09}.

When it comes to code design for fading and MIMO wiretap channels, an error probability criterion was used in several prior works \cite{BO_TComm,KHHV,KOO15}, while information theoretic security was only addressed recently with the help of flatness factors \cite{Hamed,LVL16}. In particular, \cite{LVL16} established strong secrecy over MIMO wiretap channels for secrecy rates that are within a constant gap from the secrecy capacity.

\subsection{Main Contributions}
In this paper, we propose universal codes for compound Gaussian MIMO wiretap channels that complement the recent work reported in \cite{LVL16}. The key method is discrete Gaussian shaping and a ``direct" proof of the universal flatness of the eavesdropper's lattice. This method is similar to that used in \cite{Our} to approach the capacity of compound MIMO channels so that the present paper can be considered a companion paper of \cite{Our} for wiretap channels. Note that \cite{LVL16} used an ``indirect" proof, which was based on an upper bound on the smoothing parameter in terms of the minimum distance of the dual lattice. Besides considering different channel models (\cite{LVL16} is focused on ergodic stationary channels although it also briefly addresses compound channels), the code constructions of this paper and \cite{LVL16} are also different: the construction of \cite{LVL16} is based on a particular sequence of algebraic number fields with increasing degrees, while the algebraic construction of this work combines algebraic number fields of fixed degree and random error correcting codes of increasing lengths. The proposed construction enjoys a significantly smaller gap to secrecy capacity, as well as lower decoding complexity, than \cite{LVL16}, over compound MIMO wiretap channels.

We focus on a compound channel formed by the set of all matrices with the same white-input capacity (see \eqref{eq:BandE} for the precise model). Our lattice coding scheme universally achieves rates (in nats) up to $(C_b - C_e - n_a)^{+}$, where $C_b$ is the capacity of the legitimate channel, $C_e$ is the capacity of the eavesdropper channel, $n_a$ is the number of transmit antennas and $(x)^+ = \max\left\{x,0\right\}$. We believe the $n_a$-nat gap is an artifact of our proof technique based on the flatness factor, which may be removed by improving the flatness-factor method. This is left as an open problem for future research.

For this special compound model, we also show how to extend the analysis in order to accommodate number-of-antenna mismatch, \emph{i.e.}, security is valid \textit{regardless} of the number of antennas at the eavesdropper\footnote{Previous works \cite{BO_TComm,LVL16} required that the number of the eavesdropper's antennas be greater than or equal to $n_a$.}. This is a very appealing property, since the number of receive antennas of an eavesdropper may be unknown to the transmitter.

We present two techniques to prove universality of the proposed lattice codes. The first technique is based on Construction A (see Sect. \ref{ConstructionA} for the definition) and the usual argument for compound channels \cite{RootVarayia1968,Loyka16Compound}, which combines fine quantization of the channel space with mismatch encoding for quantized states. This method is a generic proof of the existence of good codes which potentially incurs large blocklengths and performance loss. The second technique is based on algebraic lattices and assumes that the codes admit an ``algebraic reduction'' and can absorb the channel state. In fact, any code which is good for the \textit{Gaussian} wiretap channel can be coupled with this second technique, as long as it also possesses an additional algebraic structure (for precise terms see Definition \ref{def:algRed}). It is inspired by previous works on algebraic reduction for fading and MIMO channels \cite{GhayaViterboJC}, \cite{LuzziGoldenCode}, which are revisited here in terms of secrecy.

\subsection{Relation to Previous Works}

An idea of approaching the secrecy capacity of fading wiretap channels using nested lattice codes was outlined in \cite{LingISTC16}. Code construction for compound wiretap channels has been further developed in \cite{CLB-IZS18}, which leads to the current work where proof details are given.

The technique for establishing universality of the codes in \cite{SL15} over the compound MIMO channel with (uncountably) infinite uncertainty sets consists of quantizing the channel space and designing a (random Gaussian) codebook for the quantized channels. This method is similar to the proof of Theorem \ref{thm:achievableRatesFinal} in the present paper.

Compound MIMO channels \textit{without} secrecy constraints have been considered earlier in \cite{RootVarayia1968,Loyka16Compound,ShiWesel07} for random codebooks. Lattice codes are shown to achieve the optimal diversity-multiplexing tradeoff for MIMO channels in  \cite{HeshamElGamal04}. More recently it was proven that precoded integer forcing \cite{Ordentlich15} achieves the compound capacity up to a gap, while algebraic lattice codes \cite{Our} achieve the compound capacity with ML decoding and a gap to the compound capacity of MIMO channels with reduced decoding complexity. As mentioned above, some techniques (generalized Construction A and channel quantization)
of this paper are similar to those used in \cite{Our}.

\subsection{Organization}
The technical content of this paper is organized as follows. In Section \ref{sec:preliminaries} we discuss the main problem and notions of security. In Section \ref{sec:correlated}, we introduce the main notation on lattices and discrete Gaussians, stating generalized versions of known results for correlated Gaussian distributions. In Section \ref{sec:IV} we give an overview of the main coding scheme and analyze the information leakage and reliability. The proof of universality, however, is postponed until Section \ref{sec:universallyFlat}, where we show that lattice codes can achieve vanishing information leakage under semantic security through the two aforementioned techniques. Section \ref{sec:discussion} concludes the paper with a discussion of other compound models and future work.

\subsection{Notation}

Matrices and column vectors are denoted by upper
and lowercase boldface letters, respectively. For a matrix $\mathbf{A}$, its Hermitian transpose, inverse, determinant and trace are denoted by $\mathbf{A}^{\dag}$, $\mathbf{A}^{-1}$, $|\mathbf{A}|$ and $\mathrm{tr}(\mathbf{A})$, respectively. We denote the Frobenius norm of a matrix by $\left\|\mathbf{A} \right\|_F \triangleq \sqrt{\mbox{tr}(\mathbf{A}^\dagger \mathbf{A})}$ and the spectral norm (\emph{i.e.}, $2$-norm) by $\left\|\mathbf{A} \right\| \triangleq \sqrt{\lambda_1}$, where $\lambda_1$ is the largest eigenvalue of $\mathbf{A}^{\dagger} \mathbf{A}$. $\mathbf{I}$ denotes the identity matrix. We write $\mathbf{A}\succeq\mathbf{0}$ for a symmetric matrix $\mathbf{A}$ if it is positive semi-definite. Similarly, we write $\mathbf{A}\succeq\mathbf{B}$
if $(\mathbf{A}-\mathbf{B})\succeq\mathbf{0}$. We use the standard asymptotic notation $%
f\left( x\right) =O\left( g\left( x\right) \right) $ when $\lim\sup_{x\rightarrow
\infty}|f(x)|/g(x) < \infty$ , $%
f\left( x\right) =o\left( g\left( x\right) \right) $ when $\lim_{x\rightarrow
\infty}f(x)/g(x) =0$,
$%
f\left( x\right) =\Omega\left( g\left( x\right) \right) $ when $\lim\inf_{x\rightarrow
\infty}f(x)/g(x) > 0$,
and $%
f\left( x\right) =\omega\left( g\left( x\right) \right) $ when $\lim_{x\rightarrow
\infty}f(x)/g(x) =\infty$. Finally, in this paper, the logarithm is taken with respect to base $e$ (where $e$ is the Neper number) and information is measured in nats.

 \section{Problem Statement}
\label{sec:preliminaries}
Consider the following wiretap model. A transmitter (Alice) sends information through a MIMO channel to a legitimate receiver (Bob) and is eavesdropped by an illegitimate user (Eve). The channel equations for Bob and Eve read:
\begin{equation}\label{eq:block-fading}\begin{split}
\underbrace{\mathbf{Y}_b}_{n_b\times T}&=\underbrace{\mathbf{H}_b}_{n_b\times n_a}\underbrace{\mathbf{X}}_{n_a \times T}+\underbrace{\mathbf{W}_b}_{n_b\times T} \\ \underbrace{\mathbf{Y}_e}_{n_e\times T}&=\underbrace{\mathbf{H}_e}_{n_e\times n_a}\underbrace{\mathbf{X}}_{n_a \times T}+\underbrace{\mathbf{W}_e}_{n_e\times T},
\end{split}
\end{equation}
where $n_a$ is the number of transmit antennas, $n_{b}$ ($n_{e}$, resp.) is the number of receive antennas for Bob (Eve, resp.), $T$ is the coherence time, and $\mathbf{W}_b$ ($\mathbf{W}_e$, resp.) has circularly symmetric complex Gaussian i.i.d. entries with variance $\sigma_b^2$ ($\sigma_e^2$, resp.) per complex dimension. We can vectorize \eqref{eq:block-fading} in a natural way:
\begin{equation}\label{eq:block-fading-vec}\begin{split}
\underbrace{\mathbf{y}_b}_{n_bT \times 1}&=\underbrace{\mathcal{H}_b}_{n_bT \times n_aT}\underbrace{\mathbf{x}}_{n_aT \times 1}+\underbrace{\mathbf{w}_b}_{n_bT\times1} \\ \underbrace{\mathbf{y}_e}_{n_eT \times 1}&=\underbrace{\mathcal{H}_e}_{n_eT \times n_aT}\underbrace{\mathbf{x}}_{n_aT \times 1}+\underbrace{\mathbf{w}_e}_{n_eT\times1},
\end{split}
\end{equation}
where $\mathcal{H}_b$ and $\mathcal{H}_e$ are the block diagonal matrices
$$\mathcal{H}_{b} = \mathbf{I}_{T} \otimes \mathbf{H}_b = \left(\begin{array}{cccc} \HH_b & & & \\ & \HH_b & & \\ & & \ddots & \\ & & & \HH_b \end{array}\right),$$
$$ \mathcal{H}_{e} = \mathbf{I}_{T} \otimes \mathbf{H}_e = \left(\begin{array}{cccc} \HH_e & & & \\ & \HH_e & & \\ & & \ddots & \\ & & & \HH_e \end{array}\right).$$
For convenience, we denote the transmit signal-to-noise ratio (SNR) in Bob and Eve's channels by
\begin{equation*}
\rho_b \triangleq \frac{P}{\sigma_b^2} \mbox{ and } \rho_e \triangleq \frac{P}{\sigma_e^2},
\end{equation*}
respectively, where $P$ is the power constraint, \emph{i.e.}, the transmitted signal satisfies $\mathbb{E}[\mathbf{x}^{\dagger}\mathbf{x}]\leq n_a T P$.

We assume that the channel realizations $(\mathbf{H}_b,\mathbf{H}_e)$ are \textit{unknown} to Alice but belong to a compound set $\mathcal{S}=\mathcal{S}_b \times \mathcal{S}_e \in \mathbb{C}^{n_b \times n_a} \times \mathbb{C}^{n_e \times n_a}$. From the security perspective, we further make the conservative assumption that Eve knows both $\mathbf{H}_b$ and $\mathbf{H}_e$. Under this general scenario the (strong) secrecy capacity is bounded by \cite{SL15}:
$$C_s \geq \max_{\RR} \min_{\HH_b,\HH_e} \left(\log |\II
	+\sigma_b^{-1}\HH_b^\dagger\HH_b \RR | - \log \left|\II
	+\sigma_e^{-1}\HH_e^\dagger\HH_e \RR \right|  \right)^{+},$$
where the minimum is over all realizations in $\mathcal{S}$ and the maximum over the matrices $\mathbf{R} \succeq 0$ such that $\text{tr}(\mathbf{R}) \leq n_a P$. Suppose that $\mathcal{S}_b$ and $\mathcal{S}_e$ are the set of channels with the same isotropic mutual information, \emph{i.e.},
\begin{equation}
\begin{split}
\mathcal{S}_b &= \left\{\HH_b \in \mathbb{C}^{n_b \times n_a } : |\II
+\rho_b \HH_b^\dagger\HH_b |= e^{C_b} \right\}, \\ \mathcal{S}_e &= \left\{\HH_e \in \mathbb{C}^{n_e \times n_a } :\left|\II
+\rho_e\HH_e^\dagger\HH_e \right|= e^{C_e} \right\},
\end{split}
\label{eq:BandE}
\end{equation}
for fixed $C_b,C_e\geq0$. In this case, the bound gives $C_s \geq (C_b - C_e)^{+}$. The worst case is achieved by taking a specific ``isotropic'' realization $\mathbf{H}_b^\dagger \mathbf{H}_b = \alpha_b \II$, $\mathbf{H}_e^\dagger \mathbf{H}_e = \alpha_e \II$, where $\alpha_b$ and $\alpha_e$ are such that $\mathbf{H}_b$ and $\mathbf{H}_e$ belong to $\mathcal{S}_b$ and $\mathcal{S}_e$, respectively. From this we conclude that $C_s = C_b-C_e$. The goal of this paper is to construct universal lattice codes that approach the secrecy capacity $C_s$ with \textit{semantic} security. As a corollary, the semantic security capacity and the strong secrecy capacity of the compound set $\mathcal{S}_b \times \mathcal{S}_e$ coincide.

A practical motivation to consider the compound model \eqref{eq:BandE} is the following. Firstly, notice that the secrecy capacity is the same if we replace the equality in the definition of $\mathcal{S}_b$ and $\mathcal{S}_e$ with upper/lower bounds; more precisely the secrecy capacity of the channel with compound set $\overline{\mathcal{S}}_e \times \overline{\mathcal{S}}_b$, where
\begin{equation}
\begin{split}
\overline{\mathcal{S}}_b &= \left\{\HH_b \in \mathbb{C}^{n_b \times n_a } :|\II
+\rho_b \HH_b^\dagger\HH_b |\geq e^{C_b} \right\},\\ \overline{\mathcal{S}}_e &= \left\{\HH_e \in \mathbb{C}^{n_e \times n_a } :|\II
+\rho_e\HH_e^\dagger\HH_e |\leq e^{C_e} \right\},
\end{split}
\label{eq:BandEbar}
\end{equation}
is the same as for $\mathcal{S}_e \times \mathcal{S}_b$. Note that the sets $\mathcal{S}_b$, $\mathcal{S}_e$ and $\overline{\mathcal{S}}_e$ are compact whereas $\overline{\mathcal{S}}_b$ is not. In other words, universal codes are robust, in the sense that only a lower bound on the legitimate channel capacity and an upper bound on the eavesdropper channel are needed. From the security perspective, this is a safe strategy in the scenario where the capacities are not known precisely. Even if Bob and Eve's channels are random, an acceptable secrecy-outage probability can be guaranteed by setting $C_b$ and $C_e$ properly. Then, the problem still boils down to the design of universal codes for the compound model \eqref{eq:BandE}.

\subsection{Notions of Security}
A secrecy code for the compound MIMO channel can be formally defined as follows.
\begin{defi} An $(R,R^\prime,T)$-secrecy code for a compound MIMO channel with set $\mathcal{S} = \mathcal{S}_b \times \mathcal{S}_e$ consists of
	\begin{itemize}
		\item[(i)] A set of messages $\mathcal{M}_T = \left\{1,\ldots,e^{TR}\right\}$ (the secret message rate $R$ is measured in nats and $e^{TR}$ is assumed to be an integer for convenience).
		\item[(ii)] An auxiliary (not necessarily uniform) source $U$ taking values in $\mathcal{U}_T$ with entropy $R^\prime=H(U)$.
		\item[(iii)] A stochastic encoding function $f_T: \mathcal{M}_T \times  \mathcal{U}_T \to \mathbb{C}^{n_a\times T}$ satisfying the power constraint
		\begin{equation}
			\frac{1}{T}\text{{\upshape tr}}\left(\mathbb{E}\left[f_T(m,U)^{\dagger} f_T(m,U)]\right]\right) \leq n_a P,
		\end{equation}
		for any $m \in \mathcal{M}_T$.
		\item[(iv)] A decoding function $g_T: \mathcal{S}_b \times \mathbb{R}^{n_b \times T} \to \mathcal{M}_T$ with output $\hat{m} = g_T(s_b,\mathbf{Y}_b)$.
	\end{itemize}
\end{defi}
A pair $(s_b,s_e) \in\ \mathcal{S}_b \times \mathcal{S}_e$ is referred to as a \textit{channel state} (or \textit{channel realization}).
To ensure reliability for all channel states we require a sequence of codes whose error probability for message $M$ vanishes uniformly:
\begin{equation} \mathbb{P}_{\text{err}|M} \triangleq \mathbb{P}(\hat{M} \neq M) \to 0, \forall s_b \in \mathcal{S}_b\mbox{, as } T \to \infty.
\label{eq:Reliability1}
\end{equation}
Let $p_M$ be a message distribution over $\mathcal{M}_T$. For strong secrecy, $p_M$ is usually assumed to be uniform; however, this assumption is not sufficient from the viewpoint of semantic security, which is the standard notion of security in modern cryptography.  Let $\mathbf{Y}_{e}$ be the output of the channel to the eavesdropper, who is omniscient.  The following security notions are adapted from \cite{Bellare2012, LLBS_12} and should hold in the limit $T \to \infty$:
\begin{itemize}
	\item \textit{Mutual Information Security (\textsc{MIS})}: Unnormalized mutual information
\begin{equation}\label{eq:MIS}
\mathbb{I}(M; \mathbf{Y}_{e}) \to 0
\end{equation}
	for any message distribution $p_M$ and $\mbox{for \textit{all} }s_e \in \mathcal{S}_e$.
	\item \textit{Semantic Security (\textsc{SemanticS})}: Adversary's advantage
	$$\sup_{f,p_M}\left\{ \max_{m'} \mathbb{P}(f(M) = f(m') | \mathbf{Y}_e) - \max_{m''} \mathbb{P}(f(M) = f(m'')) \right\} \to 0$$
	for any function $f$ from $M$ to finite sequences of bits in $\left\{0,1\right\}^*$, and \textit{all} $s_e \in \mathcal{S}_e$.
	\item \textit{Distinguishing Security (\textsc{DistS})}: The maximum variational distance
	$$\max_{m',m''\in \mathcal{M}_T} \mathbb{V}(p_{\mathbf{Y}_e | m'},p_{\mathbf{Y}_e | m''}) \to 0 \mbox{ for all } s_e \in \mathcal{S}_e.$$
\end{itemize}
We stress that all three notions require a sequence of codes to be \textit{universally} secure for all channel states. Treating these notions as classes, we have the inclusions $\textsc{MIS}\subseteq\textsc{SemanticS}=\textsc{DistS}$, \emph{i.e.}, the sequences of codes satisfying $\textsc{DistS}$ are the same as the ones satisfying $\textsc{SemanticS}$ and also include those satisfying $\textsc{MIS}$ \cite[Prop. 1]{LLBS_12}. Moreover, if in the above notions we require that the convergence rate is $o(1/T)$, the three sets coincide. We thus define universally secure codes as follows.

\begin{defi}\label{def:universal}
	A sequence of codes of rate $R$ is universally secure for the MIMO wiretap channel if for all $(s_b,s_e) \in \mathcal{S}$, it satisfies the reliability condition
	\eqref{eq:Reliability1} and mutual information security \eqref{eq:MIS} uniformly.
\end{defi}

Then, semantic security follows as a corollary, which is a direct consequence of established relations between $\textsc{MIS}$ and $\textsc{SemanticS}$ \cite{Bellare2012}:

\begin{cor}
	The sequence of codes given in Definition \ref{def:universal} is semantically secure for the compound MIMO wiretap channel.
\end{cor}

In what follows we proceed to construct universally secure codes for the MIMO wiretap channel using lattice coset codes.

\section{Correlated Discrete Gaussian Distributions}
\label{sec:correlated}
 In this subsection, we exhibit essential results and concepts for the definition and analysis of our lattice coding scheme.
\subsection{Preliminary Lattice Definitions}
A (complex) lattice $\Lambda$ with generator matrix $\mathbf{B}_c \in \mathbb{C}^{n \times 2n}$ is a discrete additive subgroup of $\mathbb{C}^n$ given by
\begin{equation}
\Lambda = \mathcal{L}(\mathbf{B}_c) = \left\{ \mathbf{B}_c \xx: \xx \in \mathbb{Z}^{2n} \right\}.
\label{eq:complexLattice}
\end{equation}
A complex lattice has an equivalent real lattice generated by the matrix obtained by stacking real and imaginary parts of matrix $\mathbf{B}_c$:
$$\mathbf{B}_r = \left(\begin{array}{c} \Re(\mathbf{B}_c) \\ \Im(\mathbf{B}_c)\end{array} \right) \in \mathbb{R}^{2n \times 2n}.$$

A \textit{fundamental region} $\mathcal{R}(\Lambda)$ for $\Lambda$ is any interior-disjoint region that tiles $\mathbb{C}^n$ through translates by vectors of $\Lambda$. For any $\yy, \xx \in \mathbb{C}^n$ we say that $\yy = \xx \pmod \Lambda$ iff $\yy - \xx \in \Lambda$. By convention, we fix a fundamental region and denote by $\yy \pmod \Lambda$ the unique representative $\xx \in \mathcal{R}(\Lambda)$ such that $\yy = \xx \pmod \Lambda$. The volume of $\Lambda$ is defined as the volume of a fundamental region for the equivalent real lattice, given by $V(\Lambda) = |\BB_r|.$

Throughout this text, for convenience, we also use the matrix-notation of lattice points. If $\Lambda \subset \mathbb{C}^{nT}$ is a full-rank lattice, the matrix form representation of $\xx=(x_1,\ldots,x_{nT}) \in \Lambda$ is
$$\mathbf{X}= \left( \begin{array}{cccc} x_1 & x_2 & \cdots & x_T \\ x_{T+1} & x_{T+2} & \cdots & x_{2T} \\
x_{2T+1} & x_{2T+2} & \cdots & x_{3T} \\ \vdots & \vdots & \ddots & \vdots \\  x_{(n-1)T+1} & x_{(n-1)T+2} & \cdots & x_{nT}\end{array} \right).$$
The \textit{dual} $\Lambda^*$ of a complex lattice is defined as
$$\Lambda^* = \left\{\xx \in \mathbb{C}^n : \Re\left\langle \xx, \yy\right\rangle \in \mathbb{Z} \mbox{ for all } \yy \in \Lambda\right\}.$$
\subsection{The Flatness Factor} The flatness factor has been introduced in \cite{LLBS_12}, and will be used here to bound the information leakage of information transmission of our coding scheme.

The p.d.f. of the complex Gaussian centered at $\mathbf{c} \in \mathbb{C}^n$ is defined as
$$f_{\sigma,\mathbf{c}}(\xx) = \frac{1}{(\pi \sigma^2)^n} e^{-(\xx-\mathbf{c})^{\dagger} (\xx-\mathbf{c})/\sigma^2}.$$
We write $f_{\sigma,\Lambda}(\xx)$ for the sum of $f_{\sigma,\mathbf{c}}(\xx)$ over $\mathbf{c} \in \Lambda$. The \textit{flatness factor} of a lattice quantifies the distance between $f_{\sigma,\Lambda}(\xx)$ and the uniform distribution over $\mathcal{R}(\Lambda)$ and, as we will see, bounds the amount of leaked information in a lattice coding scheme.
\begin{defi} [Flatness factor for spherical Gaussian distributions]
	For a lattice~$\Lambda$ and a parameter~$\sigma$, the flatness factor
	is defined by:
	\begin{equation*}
	\epsilon_{\Lambda}(\sigma)  \triangleq \max_{\mathbf{x} \in
		\mathcal{R}(\Lambda)}|{
		V(\Lambda)f_{\sigma,\Lambda}(\mathbf{x})-1}|
	\end{equation*}
	where $\mathcal{R}(\Lambda)$ is a fundamental region of $\Lambda$.
\end{defi}

For a complex lattice $\Lambda \subset \mathbb{C}^n$, let $\gamma_{\Lambda}(\sigma) = \frac{
	V(\Lambda)^{\frac{1}{n}}}{\sigma^2}$ be the volume-to-noise ratio (VNR). We recall the formulas of the flatness factor and smoothing parameter, adapted to complex lattices. The flatness factor can be written as \cite[Prop. 2]{LLBS_12}:
\begin{equation}
\epsilon_{\Lambda}(\sigma) =  \left(\frac{\gamma_{\Lambda}(\sigma)}{{\pi}}\right)^{{n}}{
	\Theta_{\Lambda}\left({\frac{1}{\pi\sigma^2}}\right)}-1
=\Theta_{\Lambda^*}\left({{\pi\sigma^2}}\right)-1,
\label{flatness-dual-lattice}
\end{equation}
where $\Theta_{\Lambda}$ is the \textit{theta series} of the lattice $\Lambda$.
\begin{defi} [Smoothing parameter \cite{MR07}] \label{def:smooth}
	For a lattice $\Lambda$ and $\varepsilon > 0$, the smoothing parameter is defined by the function
	$\eta_{\varepsilon}(\Lambda)=\sqrt{2\pi}\sigma$, for the smallest $\sigma>0$ such that
	$\sum_{{\bm{\lambda}^*}\in \Lambda^* \setminus \{\mathbf{0}\}} e^{-\pi^2
		\sigma^2\|{ \bm{\lambda}^*}\|^2}\leq \varepsilon$.
\end{defi}

When we have a correlated Gaussian distribution with covariance matrix $\SSigma$
\begin{equation}\label{eq:corr-Gauss}
f_{\sqrt{\SSigma}, \mathbf{c}}(\mathbf{x}) = \frac{1}{\pi^{n}|\SSigma|} \exp\left\{-(\mathbf{x}-\mathbf{c})^T\SSigma^{-1}(\mathbf{x}-\mathbf{c})\right\},
\end{equation}
the flatness factor is similarly defined.

\begin{defi} [Flatness factor for correlated Gaussian distributions]
	\begin{equation*}
	\epsilon_{\Lambda}(\sqrt{\SSigma})  \triangleq \max_{\mathbf{x} \in
		\mathcal{R}(\Lambda)}|{
		V(\Lambda)f_{\sqrt{\SSigma},\Lambda}(\mathbf{x})-1}|
	\end{equation*}
	where $\mathcal{R}(\Lambda)$ is a fundamental region of $\Lambda$.
\end{defi}

The usual smoothing parameter in Definition \ref{def:smooth} is a scalar. To extend its definition to matrices, we say $\sqrt{2\pi\SSigma} \succeq \eta_{\varepsilon}(\Lambda)$ if $\epsilon_{\Lambda}(\sqrt{\SSigma}) \leq \varepsilon$. This induces a partial order because $\epsilon_{\Lambda}(\sqrt{\SSigma_1}) \leq \epsilon_{\Lambda}(\sqrt{\SSigma_2})$ if $\SSigma_1 \succeq \SSigma_2$.

When $\mathbf{c} = 0$ we ignore the index and write $f_{\sqrt{\SSigma}, \mathbf{0}}(\mathbf{x}) = f_{\sqrt{\SSigma}}(\mathbf{x})$.
For a covariance matrix $\SSigma$ we define the generalized-volume-to-noise ratio as
$$\gamma_{\Lambda}(\sqrt{\SSigma}) =	 \frac{V(\Lambda)^{1/n}}{|\SSigma|^{1/n}}.$$
Clearly, the effect of correlation on the flatness factor may be absorbed if we use a new lattice $\frac{ \sqrt{\SSigma}} {\sigma} \cdot \Lambda$, \emph{i.e.}, 
$\epsilon_{\Lambda}({\sigma}) = \epsilon_{\frac{ \sqrt{\SSigma}}{\sigma} \cdot \Lambda}(\sqrt{\SSigma})$. From this, and from the expression of the flatness factor, we have
\begin{eqnarray*}
  \epsilon_{\Lambda}(\sqrt{\mathbf{\SSigma}}) &=& \frac{V(\Lambda)}{\pi^n |\SSigma|} \sum_{{\bm{\lambda}} \in \mathbf{\Lambda}} e^{-\bm{\lambda}^\dagger \SSigma^{-1} \bm{\lambda}} - 1 \\
   &=& \left(\frac{\gamma_{\sqrt{\SSigma^{-1}} \Lambda}(\sigma^2)}{\pi}\right)^{n} \Theta_{\sqrt{\SSigma^{-1}} \Lambda}\left(\frac{1}{\pi\sigma^2}\right)-1.
\end{eqnarray*}
In our applications, the matrix $\SSigma$ will be determined by the channel realization \eqref{eq:block-fading}. Figure \ref{fig:Flatness} shows the effect of fading on the lattice Gaussian function. A function \eqref{eq:corr-Gauss} which is flat over the Gaussian channel (corresponding to $\SSigma = \II$) need not be flat for a channel in deep fading (corresponding to an ill-conditioned $\SSigma$), in which case an eavesdropper could clearly distinguish one dimension of the signal.
\begin{figure}[!h]
	\centering
	\subfloat[$\SSigma = 0.25 \II$]{\includegraphics[scale=0.6]{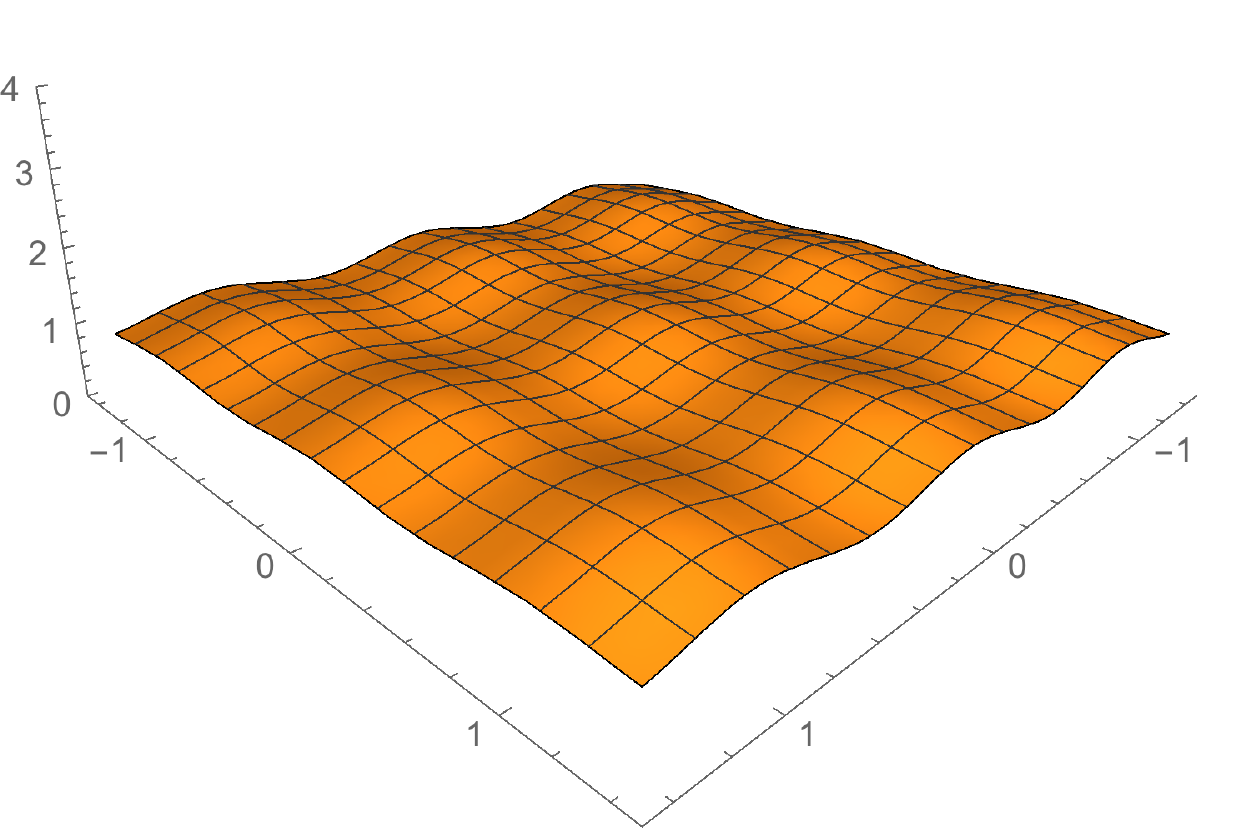}}\hfill
	\subfloat[$\SSigma=0.25\mbox{diag}(6,1/6)$]{\includegraphics[scale=0.55]{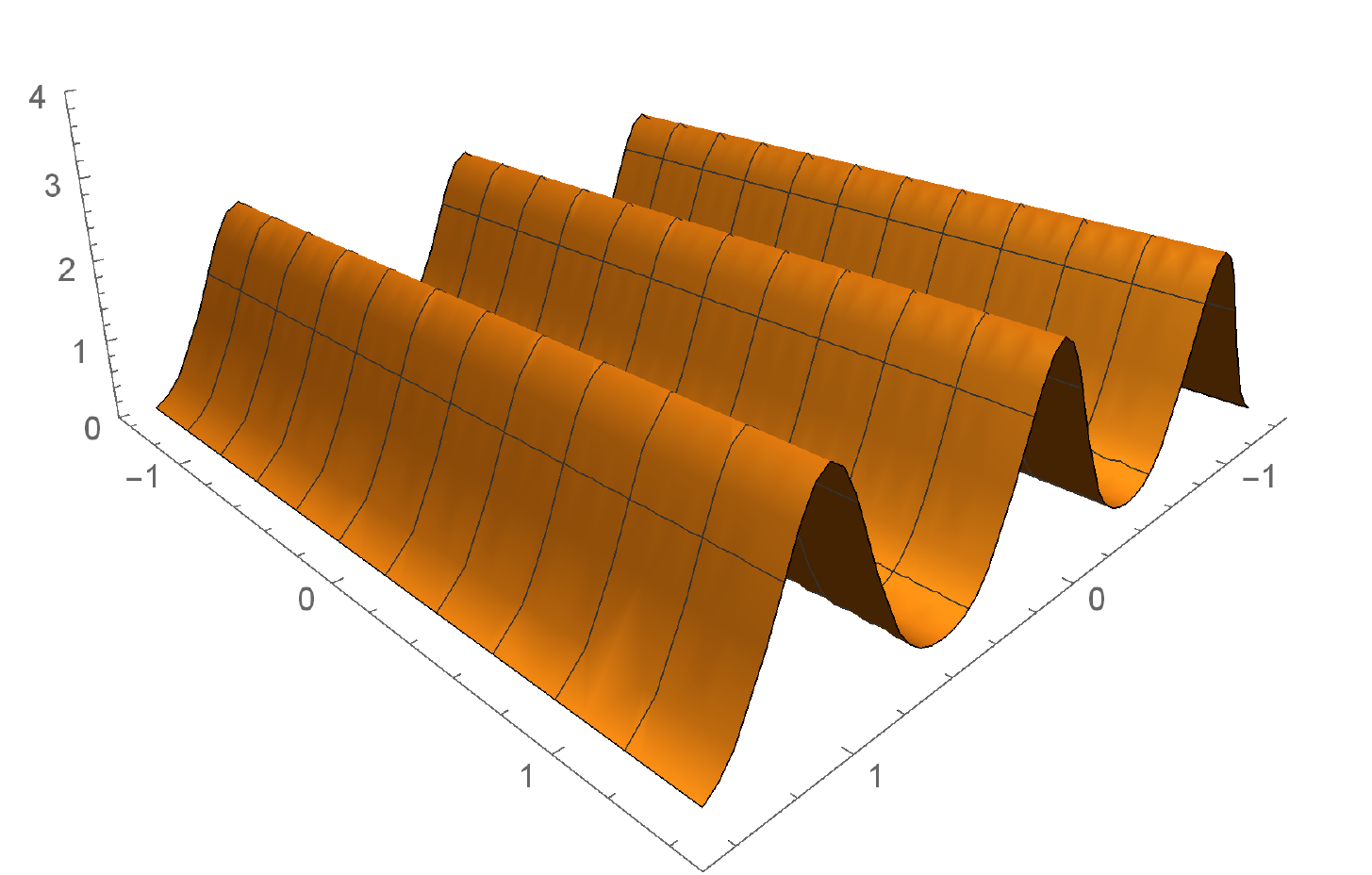}}
	\caption{Illustration of the periodic Gaussian function for the lattice $\mathbb{Z}^2$ and different covariance matrices with same determinant.}
	\label{fig:Flatness}
\end{figure}
\subsection{The Discrete Gaussian Distribution}
In order to define our coding scheme, we need a last element, which is the distribution of the sent signals. To this end, we define the \textit{discrete Gaussian distribution} $\mathcal{D}_{\Lambda+\mathbf{c},\sqrt{\SSigma}}$ as the distribution assuming values on $\Lambda+ \mathbf{c}$, such that the probability of each point $\bm{\lambda} + \mathbf{c}$ is given by
$$\mathcal{D}_{\Lambda+\mathbf{c},\sqrt{\SSigma}}(\bm{\lambda} + \mathbf{c}) = \frac{f_{\sqrt{\SSigma}}({\bm\lambda} + \mathbf{c})}{f_{\sqrt{\SSigma},\Lambda}(\mathbf{c})}.$$

Its relation to the continuous Gaussian distribution can be shown via the smoothing parameter or the flatness factor. For instance, a vanishing flatness factor guarantees that the power per dimension of $\mathcal{D}_{\Lambda+\mathbf{c},\sigma \II}$ is approximately $\sigma^2$ \cite[Lemma 6]{LLBS_12}.

The next proposition says that the sum of a continuous Gaussian and a discrete Gaussian is approximately a continuous Gaussian, provided that the flatness factor is small. The proof can be found in \cite[Appendix I-A]{LVL16}:

%

\begin{lemma}\label{lem:product}
	Given $\mathbf{x}_1$ sampled from the discrete Gaussian distribution $D_{\Lambda+\mathbf{c},\sqrt{\SSigma_1}}$ and $\mathbf{x}_2$ sampled from the continuous Gaussian distribution $f_{\sqrt{\SSigma_2}}$. Let $\SSigma_0 = \SSigma_1 + \SSigma_2$ and let $\SSigma_3^{-1} = \SSigma_1^{-1} +\SSigma_2^{-1}$. If $\sqrt{\SSigma_3} \succeq \eta_{\varepsilon}(\Lambda)$ for $\varepsilon \leq \frac{1}{2}$, then the distribution $g$ of $\mathbf{x}=\mathbf{x}_1+\mathbf{x}_2$ is close to $f_{\sqrt{\SSigma_0}}$:
	\[
	g(\mathbf{x}) \in f_{\sqrt{\SSigma_0}}(\mathbf{x})\left[ {1-4\varepsilon}, 1+4\varepsilon \right].
	\]
\end{lemma}

\section{Coding Scheme and Analysis}
\label{sec:IV}
\subsection{Overview}
\label{sec:overview}
Given a pair of nested lattices $\Lambda_e^T \subset \Lambda_{b}^T \subset \mathbb{C}^{n_aT}$ such that
$$\frac{1}{T} \log |\Lambda_b^T/\Lambda_e^T| = R,$$
the transmitter maps a message $m$ to a coset of $\Lambda_e^T$ in quotient $\Lambda_b^T/\Lambda_e^T$, then samples a point from that coset. Concretely, one can use a a one-to-one map $\phi$ such that $\phi(m) = \bm{\lambda}_m$, where $\bm{\lambda}_m$ is a representative of the coset and then samples the signal $\mathbf{x} \sim \mathcal{D}_{\Lambda_e^T+\bm{\lambda}_m,\sigma_s},$ broadcasting it to the channels. A block diagram for the transmission until the front-end receivers Bob and Eve is depicted in Figure \ref{fig:wiretapa}.

In order to find pairs of sequences of nested lattices $\Lambda_b^T$ and $\Lambda_e^T$ we employ constructions of lattices from error-correcting codes. The analysis and full construction are explained in Section \ref{sec:universallyFlat}. Essentially, the lattice $\Lambda_b^T$ controls reliability and has to be chosen in such a way that it is \textit{universally good} for the legitimate compound channel. The lattice $\Lambda_e^T$ controls the information leakage to the eavesdropper, and has to be chosen in such a way that the flatness factor vanishes universally for any eavesdropper realization (universally good for secrecy). The main result of this section is the following theorem, stating the existence of schemes with vanishing probability of error and vanishing information leakage for all pairs of realizations in the compound set $\mathcal{S}_b \times \mathcal{S}_e$.

\begin{figure*}[!htb]
	\centering
	\subfloat[Block diagram of the wiretap coding scheme.]{\includegraphics[scale=0.6]{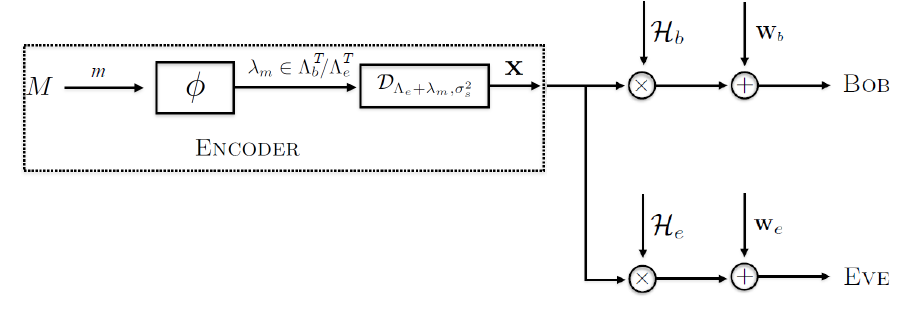}\label{fig:wiretapa}}
	\\
	\subfloat[Block diagram of Bob's receiver, where $\mathbf{F}_b$ is the \textsc{MMSE-GDFE} matrix and $\mathbf{R}_b^{-1}$ is the inverse linear operator that maps cosets of $\RR_b\Lambda_b^T/\RR_b\Lambda$ into cosets of $\Lambda_b^T/\Lambda_e^T$.]{\includegraphics[scale=0.6]{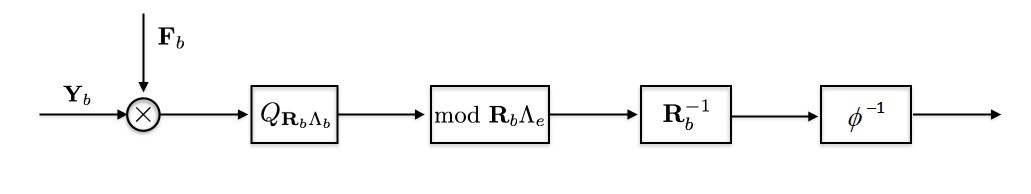}\label{fig:wiretapb}}
	\caption{Encoding and decoding over the compound wiretap channel.}	\label{fig:wiretap}
\end{figure*}

\begin{thm}
	There exists a sequence of pairs of nested lattices $(\Lambda_b^T,\Lambda_e^T)_{T=1}^{\infty}$, $\Lambda_b^T \subset\Lambda_e^T \subset \mathbb{C}^{n_aT}$ such that as $T\to \infty$, the lattice coding scheme \textit{universally} achieves any secrecy rate
	$$R < (C_b - C_e - n_a)^{+}.$$
	\label{thm:achievableRatesFinal}
\end{thm}
Moreover, we show that both the probability of error and information leakage in Theorem \ref{thm:achievableRatesFinal} vanishes uniformly for all realizations.

\subsection{The Eavesdropper Channel: Security}
For a \textit{fixed} realization $\mathbf{H}_e$, the key element for bounding the information leakage is the following lemma \cite[Lem 2]{LLBS_12}:
\begin{lemma} Suppose that there exists a probability density function $q$ taking values in $\mathbb{C}^{n_e\times T}$ such that $\mathbb{V}(p_{\mathbf{Y}_e|m}, q_{\mathbf{Y}_e}) \leq \varepsilon_T$ for all $m \in \mathcal{M}_T$. Then, for all message distributions, the information leakage is bounded as:
\begin{equation}
\mathbb{I}(M; \mathbf{Y}_e) \leq 2 n_e T \varepsilon_T R - 2 \varepsilon_T \log 2 \varepsilon_T.
\label{eq:leakage}
\end{equation}
\end{lemma}
We will show that if the distribution is sufficiently flat, then $\mathbf{Y}_e| m$ is statistically close to a multivariate Gaussian for any $m \in \mathcal{M}_T$. Let us assume for now that $\HH_e$ is an invertible square matrix (we next show how to reduce the other cases to this one). In this case, given a message $m$, we have
$$\mathcal{H}_e \mathbf{x} \sim \mathcal{D}_{\mathcal{H}_e(\Lambda_e^T + \bm{\lambda}_m), \sqrt{(\mathcal{H}_e\mathcal{H}_e^{\dagger}) \sigma_s^2}}.$$

According to Lemma \ref{lem:product}, the distribution of $\mathcal{H}_e \mathbf{x} + \mathbf{w}_e$ is within variational distance $4 \varepsilon_T$ from the normal distribution $\mathcal{N}(0,\sqrt{\SSigma_0})$, where $\varepsilon_T = \varepsilon_{\mathcal{H}_e \Lambda_e^T}(\sqrt{\SSigma_3})$ and
\begin{equation}\SSigma_0 = (\mathcal{H}_e\mathcal{H}_e^{\dagger}) \sigma_s^2 + \sigma_e^2 \mathbf{I}, \ \ \SSigma_3^{-1} =  (\mathcal{H}_e\mathcal{H}_e^{\dagger})^{-1} \sigma_s^{-2} + \sigma_e^{-2} \mathbf{I}.
\label{eq:sigma3}
\end{equation}

We thus have the following bound for the information leakage (\eqref{eq:leakage} with $\varepsilon_T$ replaced by $4\varepsilon_T$):
\begin{equation}
\mathbb{I}(M; \mathbf{Y}_e) \leq 8 n_e T \varepsilon_T R - 8 \varepsilon_T \log 8 \varepsilon_T.
\label{eq:infoLeakage}
\end{equation}	
	
Therefore, if $\varepsilon_T=\varepsilon_{\mathcal{H}_e \Lambda_e^T}(\sqrt{\SSigma_3}) = o(1/T)$, the leakage vanishes as $T$ increases \textit{for the specific realization $\mathcal{H}_e$.} To achieve strong secrecy universally, we must, however, ensure the existence of a lattice with vanishing flatness factor for \textit{all} possible $\mathbf{\SSigma}_3$. We postpone the universality discussion to Section \ref{sec:universallyFlat} where it is proven that a vanishing flatness factor is possible simultaneously for all $\HH_e \in \mathcal{S}_e$ and $\gamma_{\mathcal{H}_e \Lambda_e^T}(\sqrt{\SSigma_3}) < \pi$. This condition implies that semantic security is possible for any VNR,
\begin{equation} \gamma_{\mathcal{H}_e \Lambda_e^T}(\sqrt{\SSigma_3}) = \frac{|\mathcal{H}_ e^\dagger \mathcal{H}_e|^{1/n_a T}  {V(\Lambda_e^T)^{1/n_a T}}}{ |\SSigma_3|^{1/n_a T}} < \pi,
\label{eq:Secrecy}
\end{equation}
\begin{equation}\label{eq:vole}
V(\Lambda_e^T)^{1/n_aT} < \left| \II + \rho_e \HH_e^{\dagger} \HH_e\right|^{-1/n_a} \pi\sigma_s^2 = (\pi \sigma_s^2)e^{-C_e/n_a}.
\end{equation}
\\
\textbf{Number-of-Antenna Mismatch.} The above analysis assumed that $n_e = n_a$, \emph{i.e.}, the number of eavesdropper receive antennas is \textit{equal} to the the number of transmit antennas. Although analytically simpler, this assumption is not reasonable in practice, since we expect a compound scheme to perform well for any number of eavesdropper antennas. We show next how to reduce the other cases to the square case.

(i) $n_e < n_a$: Recall that the signal received by the eavesdropper is given in matrix form by
$$\mathbf{Y}_e = \mathbf{H}_e \mathbf{X}+ \mathbf{W}_e.$$
Let $\tilde{\mathbf{H}}_e \in \mathbb{C}^{(n_a-n_e) \times n_a}$ be a completion of $\HH_e$ such that
$$\overline{\HH}_e = \left(\begin{array}{c} \mathbf{H}_e \\ \beta \tilde{\mathbf{H}}_e \end{array}\right),$$
is a full-rank sqaure matrix and $\beta > 0$ is some small number. Let $\tilde{\mathbf{W}}_e \in \mathbb{C}^{(n_a-n_e)\times T}$ be a matrix corresponding to circularly symmetric Gaussian noise. Consider the following surrogate MIMO channel:
$$\left(\begin{array}{c} \mathbf{Y}_e \\ \tilde{\mathbf{Y}}_e \end{array}\right) = \left(\begin{array}{c} \mathbf{H}_e \\ \beta \tilde{\mathbf{H}}_e \end{array}\right) \mathbf{X}+ \ \left(\begin{array}{c} \mathbf{W}_e \\ \tilde{\mathbf{W}}_e \end{array}\right),$$
where $\tilde{\HH}_e$ is scaled so that the capacity of the new channel is arbitrarily close to the original one. Indeed for any full rank completion $\tilde{\HH}_e$, from the matrix determinant lemma, we have
\begin{equation}
|\II + \rho_e \overline{\HH}_e^\dagger \overline{\HH}_e| = |\II + \rho_e {\HH}_e^\dagger {\HH}_e| \times |\II + \beta^2 \tilde{\HH}_e (|\II + \rho_e {\HH}_e^\dagger {\HH}_e|)^{-1}\tilde{\HH}_e^{\dagger}| \geq e^{C_e}.
\end{equation}
Therefore, by letting $\beta \to 0$, the left-hand side tends to $e^{C_e}$. For any signal $\XX$, the information leakage of the surrogate channel is strictly greater than the original one. Indeed, the the eavesdropper's original channel is stochastically degraded with respect to the augmented one, thus $\mathbb{I}(M; (\mathbf{Y}_e,\tilde{\mathbf{Y}}_e)) \geq \mathbb{I}(M; \mathbf{Y}_e).$
A universally secure code for the $n_a \times n_a$ MIMO compound channel will have vanishing information leakage for the surrogate $n_a \times n_a$ channel (for \textit{any} completion) and therefore will also be secure for the original $n_e \times n_a$ channel.


(ii) $n_e > n_a$: Performing a rectangular $QR$ factorization of $\HH_{e}$ we have:
$$\HH_{e} = \mathbf{Q} \left(\begin{array}{c} \hat{\RR} \\ \mathbf{0} \end{array}\right),$$
where ${\mathbf{Q}} \in \mathbb{C}^{n_e \times n_e}$ and $\hat{\mathbf{R}} \in \mathbb{C}^{n_a \times n_a}$ are square matrices. Therefore the eavesdropper's received signal is equivalent to
\begin{eqnarray*}
  \YY_e &=& \mathbf{Q}\left(\begin{array}{c} \hat{\RR} \\ \mathbf{0} \end{array}\right)\mathbf{X} +\left(\begin{array}{c} \WW_e^{(1)} \\ \WW_e^{(2)} \end{array}\right) \\
  \iff \mathbf{Q}^{\dagger} \YY_e & = & \left(\begin{array}{c} \hat{\RR} \\ \mathbf{0} \end{array}\right)\mathbf{X} +\left(\begin{array}{c} \tilde{\WW}_e^{(1)} \\ \tilde{\WW}_e^{(2)} \end{array}\right),
\end{eqnarray*}
where the components of the noise matrices $\tilde{\WW}_e^{(1)}, \tilde{\WW}_e^{(2)}$ are i.i.d. Gaussian. The leakage is therefore the same as for the square channel $\hat{\RR}$ and a universal code will also achieve vanishing leakage for the non-square channel.

\subsection{The Legitimate Channel: Reliability}
It was shown in \cite{Our} that if $\XX \sim \mathcal{D}_{\Lambda_b^T,\sigma_s}$, then the maximum-a-posteriori (MAP) decoder for the signal $\YY_b$ is equivalent to lattice decoding of $\mathbf{F}_b \YY_b$, where $\mathbf{F}_b$ is the MMSE-GDFE matrix to be defined in the sequel. We cannot claim directly that $\XX \sim \mathcal{D}_{\Lambda_b^T,\sigma_s}$, since the message distribution in $\mathcal{M}_T$ need not be uniform. Nonetheless, we show that reliability is still possible for all individual messages.

The full decoding process is depicted in Figure \ref{fig:wiretapb}. Bob first applies a filtering matrix $\mathbf{F}_b$ so that
$$\tilde{\YY}_b = \mathbf{F}_b \YY_b = \mathbf{R}_b \XX + \mathbf{W}_{b,\eff},$$
where $\RR_b^{\dagger} \RR_b = \HH_b^{\dagger} \HH_b + \rho_b^{-1} \II$ and $\FF_b^{\dagger} \RR_b = \rho_b^{-1} \HH_b$, and the effective noise is
$$\mathbf{W}_{b,\eff} = (\FF_b \HH_b - \RR_b)\XX + \FF_b \mathbf{W}_b.$$
The next step is to decode $\tilde{\YY}_b$ in $\RR_b \Lambda_b^T$, in order to obtain $Q_{\RR_b \Lambda_b^T}(\tilde{\YY}_b),$ which is then remapped into the element of the coset $\RR_b\Lambda_b^T/\RR_b\Lambda_e^T$ through the operation $\mbox{mod }\RR_b \Lambda_e^T$. We can then invert the linear transformation associated to $\RR_b$ (notice that $\RR_b$ has full rank) in order to obtain the coset in $\Lambda_b^T/\Lambda_e^T$ and re-map it to the message space $\mathcal{M}_T$ through $\phi^{-1}$.

In the first step, from Lemma \ref{lem:product}, the effective noise $\WW_{b,\eff}$ is statistically close to a Gaussian noise with covariance:
\begin{eqnarray}
	\SSigma_{b,\eff} 	&=& \sigma_s^2 (\mathbf{F}_b\mathbf{H}_b-\mathbf{R}_b)(\mathbf{F}_b\mathbf{H}_b-\mathbf{R}_b)^\dagger + \sigma_b^2\mathbf{F}_b\FF_b^\dagger \\ \nonumber
	&=& \rho_b^{-2} \sigma_s^2 \mathbf{R}_b^{-\dagger}\mathbf{R}_b^{-1} + \sigma_b^2\rho_b^{-2}\mathbf{R}_b^{-\dagger}\mathbf{H}_b^\dagger\mathbf{H}_b\mathbf{R}_b^{-1}\\ \nonumber
	&=& \sigma_b^2 \mathbf{R}_b^{-\dagger} (\rho_b^{-1} \mathbf{I} + \mathbf{H}_b^\dagger\mathbf{H}_b)\mathbf{R}_b^{-1}= \sigma_b^2\mathbf{I}.
	\label{eq:Variance}
\end{eqnarray}
provided that $\varepsilon_{(F_b\HH_b-\RR_b)\Lambda_e^T}(\SSigma_{b,\text{inv}})$ is small, where
\begin{equation}
\SSigma_{b,\text{inv}}^{-1} = (\sigma_s^2 (\mathbf{F}_b\mathbf{H}_b-\mathbf{R}_b)(\mathbf{F}_b\mathbf{H}_b-\mathbf{R}_b)^\dagger)^{-1} + (\sigma_b^2\mathbf{F}_b\FF_b^\dagger)^{-1}.
\end{equation}
The probability of error given \textit{any} message $m$ is thus bounded by
\begin{equation}
\mathbb{P}_{\text{err}|m} \leq \left(1+4\varepsilon_{(F_b\HH_b-\RR_b)\Lambda_e^T}(\SSigma_{b,\text{inv}})\right) P(\tilde{\WW}_{b,\eff} \notin \mathcal{V}(\mathbf{R}_b \Lambda_b^T)),
\end{equation}
where each entry of $\tilde{\WW}_{b,\eff}$ is i.i.d. normal  with variance $\sigma_b^2$. Therefore, if we guarantee that $\varepsilon_{(F_b\HH_b-\RR_b)\Lambda_e^T}(\SSigma_{b,\text{inv}})$ is bounded and if we choose a universally good lattice, the probability vanishes for all possible $\RR_b$. This is possible \cite{Our} provided that
\begin{equation}\gamma_{\RR_b\Lambda_b^T}(\sigma_b) > \pi e,
\label{eq:Reliability}
\end{equation}
namely,
\begin{equation}\label{eq:volb}
V(\Lambda_b^T)^{1/n_aT} > \left| \II + \rho_b \HH_b^{\dagger} \HH_b\right|^{-1/n_a} \pi e\sigma_s^2 = (\pi e \sigma_s^2)e^{-C_b/n_a}.
\end{equation}

However, the evaluation of $\SSigma_{b,\text{inv}}$ is cumbersome and implies an extra condition for the flatness of $\Lambda_e^T$. Next we show, instead, how to circumvent this problem by using the fact that that the effective noise is ``asymptotically'' sub-Gaussian with covariance matrix $\sigma_b^2 \II$. We say that a centred random vector $\ww \in \mathbb{R}^n$ is sub-Gaussian with (proxy) parameter $\sigma$ if
$$\log E[e^{t \left\langle\ww,\uu \right\rangle}] \leq \frac{t^2 \sigma^2}{2} $$
for all $t \in \mathbb{R}$ and all unit norm vectors $\uu \in \mathbb{R}^n$.

\begin{lemma}[\cite{LVL16}]
	Let $\xx$ be a random vector with distribution $\mathcal{D}_{\Lambda_e^T+\bm{\lambda}_m,\sigma_s}$, and let $\varepsilon^\prime = \varepsilon_{\Lambda_e^T}\left(\sigma_s\right).$ For any matrix $\mathbf{A}$ and any vector $\uu \in \mathbb{C}^{n_b T}$, we have:
	$$E[e^{\Re\left\{\xx^\dagger \mathbf{A} \uu \right\}}] \leq \left(\frac{1+\varepsilon^{\prime}}{1-\varepsilon^{\prime}}\right) e^{\frac{\sigma_s^2}{4} \left\| \mathbf{A}\uu \right\|^2}.$$
\end{lemma}
Notice that the average power per dimension of a sub-Gaussian random variable is always less than or equal to its parameter $\sigma_s^2$. Moreover, the sum of two sub-Gaussians is also a sub-Gaussian (for more properties, the reader is referred to \cite{LVL16}). The above lemma, along with \eqref{eq:Variance}, allows us to establish that $\WW_{b,\eff}$ is almost sub-Gaussian with parameter $\sigma_b^2$. Therefore, as long as $\varepsilon^{\prime} \approx 0$ the probability of error tends to zero if we choose $\Lambda_b^T$ to be universally AWGN-good.

\subsection{Proof of Theorem \ref{thm:achievableRatesFinal}: Achievable Secrecy Rates}
From the previous subsections, semantic security is achievable if $\Lambda_b^T$ and $\Lambda_e^T$ satisfy:
\begin{enumerate}
	\item Reliability \eqref{eq:Reliability}: $\gamma_{\RR_b\Lambda_b^T}(\sigma_b) >\pi e$
	\item Secrecy \eqref{eq:Secrecy}: $\gamma_{\mathcal{H}_e \Lambda_e^T}(\sqrt{\SSigma_3}) < \pi$
	\item Sub-Gaussianity of equivalent noise and power constraint: $\varepsilon_{\Lambda_e^T}(\sigma_s) \to 0$.
	\end{enumerate}
From \label{eq:vole} and \eqref{eq:volb}, the first two conditions can be satisfied for rates up to
$$\log |\II + \rho_b \HH_b^{\dagger} \HH_b| - \log |\II + \rho_e \HH_e^{\dagger} \HH_e| - n_a$$
nats per channel use, but the last conditions may, \textit{a priori}, limit these rates to certain SNR regimes. Fortunately, if condition $2)$ is satisfied, we automatically satisfy the condition for $\varepsilon_{\Lambda_e^T}(\sigma_s) \to 0$, since
$$\frac{V(\Lambda_e^T)^{1/n_aT}}{\sigma_s^2} \leq \frac{V(\Lambda_e^T)^{1/n_aT}}{e^{-C_e/n_a} \sigma_s^2} < \pi.$$
Therefore, if $(\Lambda_b^T, \Lambda_e^T)$ is a sequence of nested lattices, where
\begin{enumerate}
\item $\Lambda_b^T$ is universally good for the compound channel with set $\mathcal{S}_b$,
\item $\Lambda_e^T$ is universally secure for the compound channel with set $\mathcal{S}_e$,
\end{enumerate}
then nested lattice Gaussian coding achieves any secrecy rate up to
$$R \leq (C_b - C_e - n_a)^{+}.$$

The \textit{existence} of such nested pairs is proved subsequently in Section \ref{sec:universallyFlat} and Appendix \ref{ref:appSimultaneous}, which concludes the proof of Theorem \ref{thm:achievableRatesFinal}.

In fact using a method in \cite{HB14} we can further reduce the gap to approximately $n_a \log (e/2)$. We conjecture that this gap can be completely removed with tighter bounds for the variational distance between the discrete and continuous Gaussians. This is left as an open question.
%
\begin{rmk}
	Theorem \ref{thm:achievableRatesFinal} is also a slight improvement on the main result of \cite[Theormm 5]{LLBS_12} in the sense that one of the conditions on the SNR of Bob ($\rho_b > e$) is not needed any longer. Indeed, for the Gaussian channel, $n_a = 1$ and the SNR condition for non-zero secrecy rates is $C_b > C_e + 1$, which is equivalent to
	$$\frac{1+\rho_b}{1+\rho_e} > e.$$
\end{rmk}
\section{Universally Flat Gaussians}
\label{sec:universallyFlat}
The results in the previous section require the existence of sequences of lattices which are universally good for the wiretap channel. More specifically, we need a sequence $\Lambda_b^T$ which is universally AWGN-good and a sequence $\Lambda_e^{T}$ whose leakage vanishes for \textit{all} channel realizations of the eavesdropper. The first condition was studied in \cite{Our}, where it was shown, through a compactness argument, that random lattices are universal. In this section we deal with the second condition and prove the existence of lattices $\Lambda_e^{T}$ which are universally good for secrecy of the MIMO channel.

Two methods are provided to establish the main result. The first method relies solely on random lattice coding arguments and achieves secrecy capacity up to a gap of $n_a$ nats per channel use. The second method is based on algebraic reductions and exhibits a larger gap (by a factor of $\omega(n_a\log n_a)$) to capacity, but has the appealing property of reducing the problem to the one of constructing secrecy-good lattices for the \textit{AWGN channel}, making it potentially more useful in practice.
\subsection{Construction A}{\label{ConstructionA}}
Construction A (or ``mod-$p$'') lattices are certainly the simplest choice for constructing pairs of nested lattices, however generalizations based on algebraic lattices may offer greater flexibility in the code design, which could be leveraged to obtain better decoding complexity, diversity, or other parameters. Moreover, the coding scheme in Section \ref{sec:algebraicApproach} entails an extra condition on the ensemble, which can be satisfied by assuming an algebraic structure. A general ``flexible'' construction can be defined via ``generalized reductions''. Let $\psi: \Lambda_{\text{base}} \to \mathbb{F}_p^T$ be a surjective homomorphism from a base lattice $\Lambda_{\text{base}}$ of complex dimension $N \geq T$ to the vector space $\mathbb{F}_p^T$ (also referred to as a \textit{reduction}). Define the lattice $\Lambda(\mathcal{C})$ as the pre-image of a linear code $\mathcal{C}$,
$$\Lambda(\mathcal{C}) = \psi^{-1}(\mathcal{C}).$$
If $\mathcal{C}$ has length $T$ and dimension $k$, the volume of $\Lambda(\mathcal{C})$ equals to $p^{T-k} V(\Lambda_{\text{base}})$. For instance if $N=T$, $\Lambda_{\text{base}} = \mathbb{Z}[i]^T$ the mapping $\psi$ is the reduction modulo $p$:
\begin{equation}
\begin{split}
&\psi(a_1+b_1 i, a_2+b_2 i, \ldots, a_T + b_T i) =\\
& (a_1 \,\, (\text{mod } p), b_1 \,\, (\text{mod } p), a_2 \,\, (\text{mod } p), b_2 \,\, (\text{mod } p), \cdots, \\
& a_T \,\, (\text{mod } p), b_T \,\, (\text{mod } p) ),
\end{split}
\end{equation}
 we recover an analogue of Loeliger's (mod-$p$) Construction A \cite{Loeliger}. In this case we obtain a nested lattice beween $\mathbb{Z}[i]^T$ and $p\mathbb{Z}[i]^T$. More refined ``direct'' constructions can be obtained by using number theory and prime ideals of $\mathbb{Z}[i]$. For instance, if $\Lambda_{\text{base}}$ is the embedding of the ring of integers of a number field and $\psi$ is the reduction modulo a prime ideal we can recover the constructions in \cite{Our}.
Notice that, for this construction, if $\mathcal{C}_1 \subset \mathcal{C}_2$, we obtain two nested lattices $\Lambda(\mathcal{C}_1) \subset \Lambda(\mathcal{C}_2)$.

It was shown in \cite{Campello17} that if $\{\psi\}$ is an infinite sequence of mappings, under mild conditions\footnote{More specifically, it is required that that the sequence of lattices corresponding to the kernels of $\psi$ has a non-vanishing Hermite parameter.} the ensemble of lattices averaged over all linear codes $\mathcal{C}$ of same dimension $k$ satisfies the \textit{Minkowski-Hlawka theorem}, namely:
$$\lim_{p \to \infty} \mathbb{E}_{\mathcal{C}}\left[\sum_{\xx \in \beta \Lambda(\mathcal{C}) \nozero} f(\xx)\right] = V^{-1} \int_{\mathbb{C}^{N}} f(\xx),$$
where $\beta = V^{1/2N} (p^{T-k} V(\Lambda))^{1/2N}$ is a constant so that all lattices have volume $V$. The result holds for any integrable function $f$ which decays sufficiently fast (in particular any function upper bounded by a constant times $1/(\left\|\mathbf{x}\right\| + 1)^{2N+\delta}$ for some $\delta > 0$). Clearly the Gaussian probability density function satisfies this restriction.

\subsection{Lattices Which Are Good for Secrecy}
In what follows we will apply the generalized version of Construction A to construct a sequence of lattices $\Lambda_e^T$ which is good for secrecy, \emph{i.e.}, which has vanishing flatness factor for all eavesdropper channel realizations. As usual, $T$ will denote the blocklength (cf. Equation \eqref{eq:block-fading}), $N$ will be set to $n_aT$ (the complex dimension of the coding lattice) and $k < T$ is any positive integer.

Using the above Minkowski-Hlawka theorem, there exists an ensemble of lattices $\mathbb{L}$ of volume $V$ such that
\begin{equation}
\mathbb{E}_{\mathbb{L}}\left[\sum_{\xx \in \Lambda \nozero}f(\xx)\right] \leq V^{-1} \int_{\mathbb{C}^{n_a T}} f(\xx)\mathbf{d} \xx + \varepsilon,
\label{eq:MH}
\end{equation}
for any $\varepsilon > 0$ . Equation \eqref{eq:MH} implies that
$$\mathbb{E}_{\mathbb{L}}\left[\sum_{\xx \in \Lambda\nozero} e^{-\xx^\dagger \SSigma^{-1} \xx} \right] \leq  V^{-1}\pi^{n_a T} |\SSigma| + \varepsilon,$$
therefore
$$\mathbb{E}_{\mathbb{L}}\left[\varepsilon_{\Lambda}(\sqrt{\SSigma}) \right] \leq \frac{V(1+\varepsilon)}{\pi^{n_a T} |\SSigma| }.$$
Hence as long as $\varepsilon$ is bounded and $V^{1/T}/\pi |\SSigma|^{1/T}$ is bounded by a constant less than $1$, the flatness factor tends to zero exponentially in the proposed lattice coding scheme. The condition for $\varepsilon$ can be achieved, for instance, by choosing $p$ sufficiently large in Construction A.

\begin{lemma}[Universally Flat Lattice Gaussians]
	Let $\mathcal{H}_e = \HH_e \otimes \II$ and $\SSigma_3^{-1} =  (\mathcal{H}_e\mathcal{H}_e^{\dagger})^{-1} \sigma_s^{-2} + \sigma_e^{-2} \mathbf{I},$ as in Equation \eqref{eq:sigma3}.	For any $\gamma < \pi$, there exists a sequence of lattices $\Lambda_e^T \subset \mathbb{C}^{n_aT}$ with $\gamma_{\mathcal{H}_e \Lambda_e^T}(\sqrt{\SSigma_3}) \leq \gamma $ and universally vanishing flatness factor, \emph{i.e.},
	$$\lim_{T \to \infty} \varepsilon_{\mathcal{H}_e \Lambda_e^T}(\sqrt{\SSigma_3})= 0 \mbox{ for all } \HH_e \in \mathcal{S}_e.$$
	Moreover, the convergence rate is exponential, \emph{i.e.}, for all $\HH_e \in \mathcal{S}_e$, $\varepsilon_{\mathcal{H}_e \Lambda_e^T}(\sqrt{\SSigma_3}) = e^{-\Omega(T)}$.
	\label{lem:universallyFlatGaussians}
\end{lemma}
\begin{proof}
	The proof is analogous to the quantization argument for the probability of error in \cite{Our}, which, in turn follows \cite{ShiWesel07}.
	
	(i) \textit{Fixed $\HH_e$}. If $\mathbb{L}$ is a Minkowski-Hlawka ensemble with volume $V$, then

	$$\mathbb{E}_{\mathbb{L}}\left[\varepsilon_{\mathcal{H}_e \Lambda^{(T)}}(\SSigma_3) \right] \leq (1+ \varepsilon) \left(\frac{\gamma_{\mathcal{H}_e \Lambda^{(T)}}(\sqrt{\SSigma_3})}{\pi}\right)^{n_aT}$$
	which guarantees a sequence $\Lambda_e^{(T)}$ (at this point, possibly depending on $\HH_e$) with vanishing flatness factor as long as $\gamma < \pi$.
	
	(ii) \textit{Finite set}. Let $\mathcal{S}_{Q} \subset \mathcal{S}_e$ be a \textit{finite} subset of $\mathcal{S}_e$ with cardinality $Q$. We have
\begin{eqnarray*}
  \mathbb{E}_{\mathbb{L}}\left[\max_{\HH_e \in \mathcal{S}_Q} \varepsilon_{\mathcal{H}_e \Lambda_e^{(T)}}(\SSigma_3) \right] &\leq& \mathbb{E}_{\mathbb{L}}\left[\sum_{\HH_e\in \mathcal{S}_Q} \varepsilon_{\mathcal{H}_e \Lambda_e^{(T)}}(\SSigma_3) \right] \\
  =Q(1+\varepsilon)\left(\frac{\gamma_{\mathcal{H}_e \Lambda_e^{(T)}}(\sqrt{\SSigma_3})}{\pi}\right)^{n_aT} &\to& 0
\end{eqnarray*}
	which guarantees a sequence $\Lambda_e^{(T)}$ with exponentially vanishing flatness factor for any $\HH_e \in \mathcal{S}_e$.
	
	(iii) \textit{Quantization step}. By quantizing the channel space, we can extend step (ii) into a universal code for any channel in $\mathcal{S}_e$. This analysis is described in Appendix \ref{app:2}. Here we provide a sketch of the argument. Suppose $\mathcal{S}_Q$ is a $\delta$-covering for $\mathcal{S}_e$, \emph{i.e.}, for all $\HH_{e}\in \mathcal{S}_e$, there exists $\HH_{q} \in \mathcal{S}_Q$ such that $\left\| \HH_e- \HH_q \right\| \leq \delta$. From the compactness of $\mathcal{S}_e$, such a covering exists for any arbitrarily small $\delta > 0$, and the size of the covering depends only on $n_a$, which is fixed for the whole transmission. Furthermore, the theta series is a continuous function of $\HH_e$, which implies that the flatness factor in different channel realizations are also close. From this, we can choose $\delta$ independently of $T$ that guarantees that the total exponent is negative. Therefore, the flatness factor tends to zero uniformly as $T \to \infty$.
%
\end{proof}

The above proof does not rely on a specific realization but rather on the knowledge of the compact \textit{compound} set $\mathcal{S}_e$. It is reminiscent of a widely used technique in coding for compound channels (\emph{e.g.}, \cite{ShiWesel07}). Essentially, an encoder develops a code for $Q_{\delta}$ channels, where $Q_{\delta}$ is the cardinality of a good quantizer of the channel space. However the quantization $Q_{\delta}$ may increase the effective blocklength for a target information leakage. Moreover, the proof does not give us insights on how to effectively quantize $\mathcal{S}_e$, making algebraic approaches appealing in practice.

Lemma \ref{lem:universallyFlatGaussians} shows the existence of universally flat Gaussians or, in other words, the existence of a sequence of lattices $\Lambda_e^{T}$ which are good for secrecy. Recall that in our construction \ref{sec:overview}, we required $\Lambda_e^{T}$ to be nested with $\Lambda_b^{T} \supset \Lambda_e^{T}$, where $\Lambda_b^T$ is a sequence of lattices which are good for the legitimate compound channel. The existence of $\Lambda_b^T$ was proven in \cite{Our}. In Appendix \ref{ref:appSimultaneous} we argue that both conditions can be achieved by a nested pair $(\Lambda_b^T, \Lambda_e^T)$ which is the last missing part of the proof of Theorem \ref{thm:achievableRatesFinal}.
%

\subsection{Algebraic Approach}
\label{sec:algebraicApproach}

Following \cite{LuzziGoldenCode}, we now define a lattice admitting algebraic reduction.

\begin{defi}[EU Decomposition] We say that $\Lambda$ admits algebraic reduction if for any unit determinant matrix $\mathbf{A} \in \mathbb{C}^{n_a \times n_a}$ there exists a matrix decomposition of the form $\mathbf{A} = \EE \UU$, where $\EE$ and $\UU$ are also unit-determinant satisfying the following properties:
	\begin{enumerate}
		\item $\UU \Lambda = \Lambda$,
		\item $\left\| \EE^{-1} \right\|_F \leq \alpha$ for some absolute constant $\alpha$ that does not depend on $\mathbf{A}$.
	\end{enumerate}	
	\label{def:algRed}
\end{defi}

The Golden Code is one example of a lattice that admits algebraic reduction \cite{LuzziGoldenCode}. Lattices built from generalized versions of Construction A based on number fields and division algebras also admit a similar reduction (if necessary we may relax requirement 1) to include equivalence instead of equality). This property was used in \cite{Our} to achieve capacity of the infinite compound MIMO channel. Note that $\alpha$ grows with $n_a$. See \cite[Theorem 3]{Our} for an upper bound on $\alpha$ in the case of number fields, and \cite{LuzziGoldenCode} in the case of division algebras. Next, we show that an ensemble of lattices satisfying Definition \ref{def:algRed} achieves the secrecy capacity of the compound MIMO channel up to a constant gap.

Recall the following relation between the spectral norm and the Frobenius norm:
$$ \left\|\II\otimes \mathbf{A} \right\| = \left\|\mathbf{A} \right\| \leq \left\|\mathbf{A} \right\|_F,$$
for the identity matrix $\II$ of any dimension.

\begin{lemma}
	Suppose that $\Lambda \subset \mathbb{C}^{n_a T}$ is such that its dual lattice $\Lambda^*$ admits algebraic reduction. Then for $\mathbf{A} \in \mathbb{C}^{n_a \times n_a}$,
	$$\varepsilon_{\Lambda}(\sqrt{\II_T \otimes \mathbf{A}}) \leq \varepsilon_{\Lambda}\left(|{\mathbf{A}}|^{1/2n_a}/\alpha \right).$$
\end{lemma}
\begin{proof}
	From the Poisson summation formula and the expression for the flatness factor \eqref{flatness-dual-lattice}:
\begin{equation}
	\begin{split}
	\varepsilon_{\Lambda}(\sqrt{\II_T \otimes \mathbf{A}}) &= \sum_{\bm{\lambda} \in \Lambda^{*}\nozero} e^{-\pi^2 \bm{\lambda}^\dagger (\II_T \otimes \mathbf{A}) \bm{\lambda} } \\ &=\sum_{\bm{\lambda} \in \Lambda^{*}\nozero} e^{-\pi^2 \left\| \sqrt{\II_T \otimes \mathbf{A}} \bm{\lambda}\right\|^{2} } \\
	&= \sum_{\bm{\lambda} \in \Lambda^{*}\nozero} e^{-\pi^2 |\mathbf{A}|^{1/n_a} \left\| \frac{\sqrt{\II_T \otimes \mathbf{A}}}{\left(\sqrt{| \mathbf{A}|}\right)^{1/n_a}} \,\, \bm{\lambda}\right\|^{2} }.
	\end{split}
\end{equation}

Upon decomposing $\frac{\sqrt{\mathbf{A}}}{\left(\sqrt{|\mathbf{A}|}\right)^{1/n_a}} = \EE \UU$ as in Definition \ref{def:algRed}, the last equation becomes
\begin{equation}
	\begin{split}
	&\sum_{\bm{\lambda} \in \Lambda^{*}\nozero} e^{-\pi^2 |\mathbf{A}|^{1/n_a} \left\|(\II \otimes \EE) \bm{\lambda}\right\|^2} \stackrel{(a)}{\leq} \sum_{\bm{\lambda} \in \Lambda^{*}\nozero} e^{-\pi^2 |\mathbf{A}|^{1/n_a} \frac{\left\|\bm{\lambda}\right\|^2}{\left\|\EE^{-1}\right\|^2}} \\
	&\stackrel{(b)}{\leq} \sum_{\bm{\lambda} \in \Lambda^{*}\nozero} e^{-\pi^2 |\mathbf{A}|^{1/n_a} \frac{\left\|\bm{\lambda}\right\|^2}{\left\|\EE^{-1}\right\|_F^2}} \stackrel{(c)}{\leq} \sum_{\bm{\lambda} \in \Lambda^{*}\nozero} e^{-\pi^2 |\mathbf{A}|^{1/n_a} \frac{\left\|\bm{\lambda}\right\|^2}{\alpha^2}} \\
	&= \varepsilon_{\Lambda}\left(|{\mathbf{A}}|^{1/2n_a}/\alpha \right),
	\end{split}
\end{equation}
where $(a)$ is due to the bound $\left\| \bm{\lambda}\right\| \leq \left\|(\II \otimes \EE)^{-1}\right\| \left\|(\II \otimes \EE)\bm{\lambda}\right\|$ and the fact that $\left\|(\II \otimes \EE)^{-1}\right\| = \left\|(\II \otimes \EE^{-1})\right\|  = \left\|\EE^{-1}\right\| $, $(b)$ is due to the inequality between the $2$-norm and the Frobenius norm and $(c)$ follows from Definition \ref{def:algRed}.
\end{proof}

Since $$\varepsilon_{\mathcal{H}_e\Lambda_e^T}(\sqrt{\SSigma_3}) =  \varepsilon_{\Lambda_e^T}(\sqrt{\SSigma}),$$
where $\SSigma^{-1} = \mathcal{H}_e^\dagger \SSigma_3^{-1} \mathcal{H}_e$ and
$\SSigma_3^{-1} =  (\mathcal{H}_e\mathcal{H}_e^{\dagger})^{-1} \sigma_s^{-2} + \sigma_e^{-2} \mathbf{I}$ is block-diagonal, we can apply the above lemma.  Therefore, if we construct an ensemble of lattices such that their duals admit algebraic reduction for some constant $\alpha > 0$, then there exist lattices with vanishing flatness factor $\varepsilon_{\mathcal{H}_e \Lambda_e^T}(\sqrt{\SSigma_3})$ provided that
\begin{equation}\varepsilon_{\Lambda_e^T}\left(\sqrt{\sigma_s^2 \alpha^{-2} e^{-C_e/n_a}}\right) \to 0.
\label{eq:secrecyGoodGaussian}
\end{equation}
 This can be achieved if:
\begin{equation}V(\Lambda_e^T)^{1/n_aT} < \pi e^{-C_e/n_a} \alpha^{-2} \sigma_s^2.
\label{eq:lastexpression}
\end{equation}
Notice that the right-hand side of \eqref{eq:lastexpression} depends only on the determinant of $\SSigma_3$ or on the capacity of the eavesdropper channel, not on any individual realization. For this condition to hold, we only need a sequence of secrecy-good lattices for a surrogate eavesdropper channel with smaller noise variance (by a factor ${\alpha^{-2}}$). Therefore, by combining \eqref{eq:volb} and \eqref{eq:lastexpression}, we arrive at the the following result:

\begin{thm}
Let $(\Lambda_b^{T}, \Lambda_e^{T})$ be a sequence of nested lattices where: (i) $\Lambda_b^{T}$ is universally good for the compound MIMO channel and (ii) $\Lambda_e^{T}$ satisfies Definition \ref{def:algRed} and is secrecy good \textbf{for the AWGN channel} (Condition \eqref{eq:secrecyGoodGaussian}). Then nested lattice Gaussian coding achieves any secrecy rate up to
$$R \leq (C_b - C_e - n_a -2n_a\log (\alpha))^{+}.$$
\label{thm:acheivableRatesAlgebraic}
\end{thm}

Notice that the gap has a different nature than the one in the previous subsection. It consists of two parts: $n_a$ due to the same restriction on the flatness factor in Theorem \ref{thm:achievableRatesFinal}, and $\log \alpha$ due to algebraic reduction. Although we have conjectured that the gap in Theorem \ref{thm:achievableRatesFinal} can be essentially removed, this is not the case for $\log \alpha$ in Theorem \ref{thm:acheivableRatesAlgebraic}. Indeed, since $\alpha$ cannot be smaller than $\sqrt{n_a}$ \cite[Theorem 3]{Our}, this gap is always larger than $n_a\log{n_a}$. However, the code construction can be reduced to the problem of finding good lattices for the Gaussian wiretap channel (with some additional algebraic structure), making the design potentially more practical.

Notice also that this strategy is closely related to the ``decoupled design" for compound MIMO channels \cite[Sect. VI]{Our}. Both strategies can indeed be combined, \emph{i.e.}, Bob's code can also benefit from algebraic reduction. In this case both the original channel decoder and the code design can be greatly simplified, at the cost of an extra gap (\emph{i.e.},, an extra factor $2n_a\log (\alpha)$) to the compound capacity.

\section{Discussion}\label{sec:discussion}

In this paper, we have presented a construction of nested lattice codes to achieve the secrecy capacity of compound MIMO wiretap channels, up to a gap equal to the number of transmit antennas.  Compared to \cite{LVL16}, the construction in this work is not only more practical, but also enjoys a smaller gap. With algebraic reduction, further simplification has been made, at the cost of an extra gap to the secrecy capacity. Interestingly, the algebraic approach reaffirms the important role of the dual lattice of $\Lambda_e^T$ in wiretap channels, firstly discovered in \cite{LVL16}.

\paragraph{Encoding and decoding} Encoding and decoding are not much different from those of lattice codes for compound MIMO channels in \cite{Our}. The generalized Construction A employed in this paper may be viewed as a concatenated code, where the inner code is a lattice with some desired properties, while the outer code is an error correction code. Therefore, decoding can be run successively, which greatly reduces the decoding complexity. As for encoding, the discrete Gaussian shaping can be facilitated by choosing a nice base lattice, \emph{e.g.}, a rotated $\mathbb{Z}^{n_a}$ lattice whose Gaussian shaping is easy. There are highly efficient algorithms for Gaussian shaping over specific lattices \cite{Antonio16}, but more research is needed for Gaussian shaping over generalized Construction A. Practical implementation of the proposed codes is left as future work.

\paragraph{Comparison to other compound models} When the channel $\HH_b$ is known and the eavesdropper channel has bounded norm, \cite{SL15} has shown that the eavesdropper's worst channel is also isotropic. In this case the capacity can be achieved by decomposing the channels into different independent substreams with appropriate power, and applying independent coding for the Gaussian channel. This is also the case when $\HH_b$ has a linear uncertainty. In these cases, a combination of correct power allocation and a similar argument to Lemma \ref{lem:universallyFlatGaussians} shows that semantic secrecy is also achievable by random lattice codes. On the other hand, the algebraic approach (Theorem \ref{thm:acheivableRatesAlgebraic}) heavily relies on the fact that the channels in $\mathcal{S}_e$ have the same white-input mutual information.

\paragraph{Finite-length performance}

The results of this work are based on asymptotic analysis as $T\to \infty$. The practical performance of the proposed universal codes at finite block lengths warrants an investigation. In particular, how large $T$ is required to approach the promised gap in practice? For given $T$, how far do practical codes perform from secrecy capacity? It may be a challenging problem to design good, practical universal codes.

As a further perspective, one may consider an ``outage'' analysis of the MIMO wiretap channel in a finite blocklength regime, where the channel matrices $\HH_b$ and $\HH_e$ may be random. In other words, one may analyze the probability that the code rate $R$ exceeds the secrecy capacity. In such scenarios, we believe that lattices with the non-vanishing determinant property will be able to provide universal bounds for the outage probability. We leave it as an open problem.

\appendices

\section{Quantization of Channel Space}
\label{app:2}
In this appendix we show bounds on the flatness factor in the quantized channel space, formalizing part (iii) in the proof of Lemma \ref{lem:universallyFlatGaussians}. Instead of performing the quantization directly in the eavesdropper space $\mathcal{S}_b$, we will consider the corresponding covariance matrices. Following the notation of Lemma \ref{lem:universallyFlatGaussians}, we have:
$$\varepsilon_{\mathcal{H}_e\Lambda_e^T}(\sqrt{\SSigma_3}) =  \varepsilon_{\Lambda_e^T}(\sqrt{\SSigma}),$$
where $\SSigma^{-1} = \mathcal{H}_e^\dagger \SSigma_3^{-1} \mathcal{H}_e$ and
$\SSigma_3^{-1} =  (\mathcal{H}_e\mathcal{H}_e^{\dagger})^{-1} \sigma_s^{-2} + \sigma_e^{-2} \mathbf{I}$. Let $\Omega_e$ be the space of co-variance matrices of the form $\SSigma$, where $\HH_e$ can be any matrix in the space of eavesdropper matrices $\mathcal{S}_b$:
\begin{equation*}
\Omega_e = \left\{ \SSigma = (\mathcal{H}_e^\dagger \SSigma_3^{-1} \mathcal{H}_e)^{-1} : \mathcal{H}_e = \II_T\otimes\HH_e, \HH_e \in \mathcal{S}_e  \right\}.
\end{equation*}
By using the definition of the flatness factor, we can show the following:

\begin{lemma}
	Let $\SSigma, \bar{\SSigma} \in \Omega_e$ be two matrices satisfying $\left\|\mathbf{\SSigma}-\overline{\mathbf{\SSigma}}\right\| \leq \delta$. If $\delta$ is sufficiently small, then $\overline{\SSigma} - \delta \II$ is positive-definite and
	$$\varepsilon_{\Lambda_e^T}(\sqrt{\SSigma}) \leq \varepsilon_{\Lambda_e^T}(\sqrt{\overline{\SSigma} - \delta \II}).$$
\end{lemma}
\begin{proof}
	For any $\bm{\lambda} \in \mathbb{C}^{n_aT}$ we have $|\bm{\lambda}^\dagger (\SSigma-\overline{\SSigma}) \bm{\lambda}| \leq \left\|\bm{\lambda}\right\|^2 \delta.$ Therefore
	\begin{equation*}
	\begin{split}
	\varepsilon_{\Lambda_e^T}(\sqrt{\SSigma}) &= \sum_{\bm{\lambda} \in \Lambda^{*}\nozero} e^{-\pi^2 \bm{\lambda}^\dagger \SSigma \bm{\lambda} } \\ &= \sum_{\bm{\lambda} \in \Lambda^{*}\nozero} e^{-\pi^2 \bm{\lambda}^\dagger\left(\SSigma-\overline{\SSigma}\right) \bm{\lambda} }e^{-\pi^2 \bm{\lambda}^\dagger\overline{\SSigma} \bm{\lambda} }
	\\ &\leq \sum_{\bm{\lambda} \in \Lambda^{*}\nozero} e^{\pi^2 \delta  \bm{\lambda}^\dagger \bm{\lambda} }e^{-\pi^2 \bm{\lambda}^\dagger\overline{\SSigma} \bm{\lambda} } = \varepsilon_{\Lambda_e^T}(\sqrt{\overline{\SSigma} - \delta \II}).
	\end{split}
	\end{equation*}
	
\end{proof}
Suppose now that $\mathcal{S}_{\delta}$ is a $\delta$-quantizer for $\Omega_e$ with cardinality $Q_{\delta}$, \emph{i.e.}, for all $\SSigma$ there exists $\overline{\SSigma} \in \mathcal{S}_{\delta}$ such that $\left\|\mathbf{\SSigma}-\overline{\mathbf{\SSigma}}\right\| \leq \delta$. For any $\SSigma$ we have:
\begin{equation*}
\begin{split}
\mathbb{E}[\varepsilon_{\Lambda_e^T}(\sqrt{\SSigma})] &\leq \mathbb{E}[\varepsilon_{\Lambda_e^T}(\sqrt{\overline{\SSigma}-\delta\II})] \leq
\mathbb{E}[\max_{\overline{\SSigma} \in \mathcal{S}_{\delta}}\varepsilon_{\Lambda_e^T}(\sqrt{\overline{\SSigma}-\delta \II})] \\ &\leq
\mathbb{E}[\sum_{\overline{\SSigma} \in \mathcal{S}_{\delta}}\varepsilon_{\Lambda_e^T}(\sqrt{\overline{\SSigma}-\delta \II})] \\
&\leq Q_{\delta} (1+\varepsilon_T)\left(\frac{\gamma_{\Lambda_e^T}(\sqrt{\overline{\SSigma}-\delta\II})}{\pi}\right)^{n_aT} \\ &= Q_{\delta}(1+\varepsilon_T) f(\delta) \left(\frac{\gamma_{\Lambda_e^T}(\sqrt{{\SSigma}})}{\pi}\right)^{n_aT},
\end{split}
\end{equation*}
where
$$f(\delta) = \frac{|\SSigma|}{|\SSigma-\delta \II|} = \frac{1}{|\II - \delta \SSigma^{-1}|}.$$

The last upper bound is universal, in the sense that it does not depend on the specific realization $\HH_e$. Note that if the VNR condition is satisfied, namely $\gamma_{\Lambda_e^T}(\sqrt{{\SSigma}}) < \pi$, then the term $(\gamma_{\Lambda_e^T}(\sqrt{{\SSigma}})/\pi)^{n_aT}$ decays exponentially in $T$ with exponent given by
$$c_1 = -n_a \log (\gamma_{\Lambda_e^T}(\sqrt{{\SSigma}})/\pi).$$
From this, we obtain the bound
 \begin{equation}
 \begin{split}
 \mathbb{E}[\varepsilon_{\Lambda_e^T}(\sqrt{\SSigma})] &\leq (1+\varepsilon_T)Q_{\delta} e^{-T \left(c_1 - \frac{\log f(\delta)}{T}\right)} \\ &= (1+\varepsilon_T)Q_{\delta} e^{-T \left(c_1 - {\log |\II -\delta(\sigma_{e}^{-2}\HH_e^\dagger \HH_e + \sigma_s^{-2} \II) | }\right)},
 \end{split}
 \label{eq:boundAverageFlatnessFactor}
 \end{equation}
 which holds for any $\SSigma \in \Omega_e$.
We can therefore choose a small $\delta$ (independently of $T$) such that the total exponent is negative. Since $Q_{\delta}$ does not depend on $T$, and $\varepsilon_T$ can be made arbitrarily small, we obtain an exponential decay of the flatness factor.

\section{Simultaneous Goodness}
\label{ref:appSimultaneous}
From Section \ref{sec:universallyFlat}, the construction of universally secure codes boils down to finding a sequence of pairs of nested lattices $\Lambda_{b}^{T} \subset \Lambda_{e}^{T}$ such that
\begin{itemize}
\item $\Lambda_{b}^{T}$ has vanishing probability of error: $\mathbb{P}_{\Lambda_b}(\RR_b) \triangleq \mathbb{P}(\tilde{\WW}_{b,\eff} \notin \mathcal{V}(\RR_b\Lambda_b^T)) \to 0$ as $T \to \infty$;
\item $\Lambda_{e}^T$ has vanishing flatness factor: $\varepsilon_{\Lambda_e^T}(\sqrt{\SSigma}) \to 0$ as $T \to \infty$,
\end{itemize}
where we recall that $\tilde{\WW}_{b,\eff}$ is the effective noise, sub-Gaussian with co-variance matrix $\sigma_b^2 \II$, $\RR_b^\dagger \RR_b = \HH_b^\dagger \HH_b + \rho_b^{-1} \II$, and
$$\SSigma^{-1} = \mathcal{H}_e^\dagger \SSigma_3^{-1} \mathcal{H}_e, \mbox{with } \SSigma_3^{-1} =  (\mathcal{H}_e\mathcal{H}_e^{\dagger})^{-1} \sigma_s^{-2} + \sigma_e^{-2} \mathbf{I}.$$
First suppose that $\mathbf{R}_b$ and $\SSigma$ are fixed.
Let $\Lambda_b^T = \Lambda(\mathcal{C}_b)$ be obtained by choosing $\mathcal{C}_b$ uniformly in the set of all codes with parameters $(T,k_b,p)$. Let $\Lambda_e^T = \Lambda(\mathcal{C}_e)$ be obtained by expurgating $k_b-k_e$ columns from $\mathcal{C}_b$. With this process $\mathcal{C}_e$ will be also chosen uniformly from all $(T,k_e,p)$ codes. We have:
\begin{equation}\label{eq:simult-good}
\begin{split}&\mathbb{E}_{\mathcal{C}_b}[\max\left\{\mathbb{P}_{\Lambda_b^T}(\RR_b),\varepsilon_{\Lambda_e^T}(\sqrt{\SSigma})\right\}] \\ &\leq  \mathbb{E}_{\mathcal{C}_b}[\mathbb{P}_{\Lambda_b^T}(\RR_b) +\varepsilon_{\Lambda_e^T}(\sqrt{\SSigma}) ]\\ &= \mathbb{E}_{\mathcal{C}_b}[\mathbb{P}_{\Lambda_b^T}(\RR_b)] +\mathbb{E}_{\mathcal{C}_e}[\varepsilon_{\Lambda_e^T}(\sqrt{\SSigma}) ] \to 0.
\end{split}
\end{equation}
Convergence of both terms in the last equation is guaranteed to be exponentially fast. Indeed:

\begin{itemize}
	\item The term $\mathbb{E}_{\mathcal{C}_b}[\mathbb{P}_{\Lambda_b^T}(\RR_b)]$ tends to zero exponentially provided that $\gamma_{\RR_b}(\Lambda_b) > \pi e$, due to AWGN-goodness of $\Lambda_b^T$.
	\item The term $\mathbb{E}_{\mathcal{C}_e}[\varepsilon_{\Lambda_e^T}(\sqrt{\SSigma}) ]$ tends to zero exponentially provided that $\gamma_{\Lambda_e^T}(\sqrt{\SSigma}) \to 0$, due to Appendix \ref{app:2}, Equation \eqref{eq:boundAverageFlatnessFactor}.
\end{itemize}

Furthermore, by considering the quantized channel spaces, similarly to Appendix \ref{app:2}, we conclude that the convergence is universal. Therefore, there exists a pair of lattices ${\Lambda_b^T,\Lambda_e^T}$ where $\Lambda_b^T$ is universally AWGN-good and $\Lambda_e^T$ is universally secrecy-good, and Theorem \ref{thm:achievableRatesFinal} follows.

\begin{rmk}
Although the above argument only demonstrates the existence of a pair of good lattices, it is possible to show a concentration result on the performance of the ensemble of nested lattices. Suppose some exponential bound $e^{-c T}$ on  \eqref{eq:simult-good} for some $c > 0$. Then, using Markov's inequality, we have that for the ensemble of nested lattices considered,
\begin{equation}
\mathbb{P}(\mathbb{P}_{\Lambda_b^T}(\RR_b) +\varepsilon_{\Lambda_e^T}(\sqrt{\SSigma}) > e^{-(c-c') T}) < e^{-c' T}, \ \ \forall  \ 0<c'<c.
\end{equation}
That is, with probability higher than $1-e^{-c' T}$ over the choice of $\mathcal{C}_b$, \eqref{eq:simult-good} stays below $e^{-(c-c') T}$. In other words, most of these nested lattices have a performance concentrating around
$e^{-c T}$.

\end{rmk}

\section*{Acknowledgment}


The authors would like to thank Laura Luzzi and Roope Vehkalahti for helpful discussions.

\bibliographystyle{IEEEtran}
\bibliography{block_fading}

\end{document}